\pdfoutput=1
\documentclass[sigconf,screen]{acmart}

\usepackage[normalem]{ulem}

\newcommand{\argmax}{\mathop{\mathrm{argmax\,}}}

\newcommand{\boldB}{{\boldsymbol{B}}}

\newcommand{\boldR}{{\boldsymbol{R}}}

\usepackage{bbm}
\usepackage{mathtools}
\usepackage[ruled, algo2e, vlined]{algorithm2e}

\SetCommentSty{mycommfont}

\newcommand{\bft}{\fontseries{b}\selectfont}

\PassOptionsToPackage{pdftex}{graphicx}

\AtBeginDocument{%
  \providecommand\BibTeX{{%
    \normalfont B\kern-0.5em{\scshape i\kern-0.25em b}\kern-0.8em\TeX}}}

\copyrightyear{2022}
\acmYear{2022}
\setcopyright{acmlicensed}\acmConference[WSDM '22]{Proceedings of the Fifteenth ACM International Conference on Web Search and Data Mining}{February 21--25, 2022}{Tempe, AZ, USA}
\acmBooktitle{Proceedings of the Fifteenth ACM International Conference on Web Search and Data Mining (WSDM '22), February 21--25, 2022, Tempe, AZ, USA}
\acmPrice{15.00}
\acmDOI{10.1145/3488560.3498432}
\acmISBN{978-1-4503-9132-0/22/02}

\begin{document}

\title{Enumerating Fair Packages for Group Recommendations}

\author{Ryoma Sato}
\email{r.sato@ml.ist.i.kyoto-u.ac.jp}
\affiliation{%
  \institution{Kyoto University / RIKEN AIP}
}


\begin{abstract}
Package-to-group recommender systems recommend a set of unified items to a group of people. 
Different from conventional settings, it is not easy to measure the utility of group recommendations because it involves more than one user. In particular, fairness is crucial in group recommendations. Even if some members in a group are substantially satisfied with a recommendation, it is undesirable if other members are ignored to increase the total utility. Many methods for evaluating and applying the fairness of group recommendations have been proposed in the literature. However, all these methods maximize the score and output only one package. This is in contrast to conventional recommender systems, which output several (e.g., top-$K$) candidates. This can be problematic because a group can be dissatisfied with the recommended package owing to some unobserved reasons, even if the score is high. To address this issue, we propose a method to enumerate fair packages efficiently. Our method furthermore supports filtering queries, such as top-$K$ and intersection, to select favorite packages when the list is long. We confirm that our algorithm scales to large datasets and can balance several aspects of the utility of the packages.
\end{abstract}



\maketitle

\section{Introduction}
With the use of recommender systems in diverse scenarios, a new setting called group recommendation \cite{jameson2007recommendation, yahia2009group, baltrunas2010group, berkovsky2010group} has emerged. A group recommender system recommends items to a set of users (i.e., a group). For example, PolyLens \cite{o2001polylens} recommends movies to a group of people who are planning to go to a movie together. Let's browse \cite{lieberman1999lets} recommends web pages to a group of people who are browsing web pages together. Another emerging scenario is package recommendation, in which a set of items is considered as a unified package toward a single common goal \cite{xie2010breaking, yahia2014composite, zhu2014bundle, xie2014generating}. For example, in music playlist recommendations, each package (i.e., playlist) should be consistent with respect to genres. In movie package recommendations, each package should be diverse to avoid monotony. Package recommendations are also referred to as bundle recommendations \cite{zhu2014bundle} and composite recommendations \cite{xie2010breaking}.

In this paper, we consider the combination of these two settings, i.e., package-to-group recommendations \cite{qi2016recommending, serbos2017fairness}, where a recommender system recommends packages to a group of people, and focus on the fairness problem in this setting. For example, consider the case of selecting the music to play in a car. Suppose all except one person in a car like music in the order of pop, jazz, and rock music, whereas the single remaining person likes music in the opposite order. Although a recommender system may recommend pop music to satisfy the majority to the maximum extent, it may dissatisfy the last person, who may refuse to go along with the solution. An amicable and fair solution is to play jazz music, which is liked moderately by all occupants of the car.

Existing fair package-to-group recommendation algorithms \cite{qi2016recommending, serbos2017fairness} define fairness scores of a list and determine the optimal package by maximizing the scores. These methods can compute an excellent package with respect to their fairness scores. However, a critical drawback is that they recommend only one package. Therefore, a group has only one candidate. This is in contrast to the standard setting, where several (top-$K$) candidates are presented, and users can choose their favorite items. 

To solve this problem, we propose to \emph{enumerate} all fair packages. Our proposed method outputs a set of packages, i.e., a set of sets of items. More specifically, given a threshold $\tau$, our proposed method computes all packages with at least fairness scores $\tau$. Although these sets may contain unsatisfactory packages, users can choose their most preferred package by filtering.

However, there are several technical challenges to achieving this goal. First, it was shown that the problem of maximizing fairness is NP-hard \cite{serbos2017fairness, Lin2017fairness}. The enumeration problem is more difficult than the maximization problem because if we can efficiently solve the enumeration problem, we can solve the maximization problem as well by simply adjusting the threshold $\tau$. Besides, the enumeration problem requires more detailed solutions than the maximization problem. Specifically, not only the optimal value but also the set of qualified packages needs to be provided. Therefore, the enumeration problem is also NP-hard, and it may appear impossible to efficiently solve this problem. A crucial technical contribution of this paper is that we prove that the optimization and enumeration problems are fixed-parameter tractable (FPT) \cite{downey2012parameterized}. Therefore, they can be solved if the parameter is small. Fortunately, most instances in the literature have small values. Thus, our proposed method can efficiently solve the enumeration problem. The second challenge is the existence of numerous qualified packages because of the exponential increase in the combination of packages with respect to the number of items. To solve this problem, We manipulate sets of packages and output solutions in a compressed form, namely a ZDD, instead of manipulating them one by one. ZDDs enable efficient filtering queries as well as efficient enumeration.

The contributions of this paper are as follows:
\begin{itemize}
    \item We propose enumerating all fair packages to allow users to select their favorite packages in their own way.
    \item We propose an efficient method that enumerates all fair packages exactly. Our method not only enumerates packages but also samples packages, lists important top-$K$ packages, and efficiently filters packages.
    \item We confirm the effectiveness and efficiency of our method in empirical evaluations. We show that our method can enumerate over one trillion packages in several seconds.
\end{itemize}

\noindent \uline{\textbf{Reproducibility:}} Our code is available at \url{https://github.com/joisino/fape}. The library can be installed via \texttt{pip install fape}.

\section{Preliminary}

\subsection{Notations}

For every positive integer $n \in \mathbb{Z}_+$, $[n]$ denotes the set $\{ 1, 2, \dots n \}$. $|A|$ denotes the number of elements contained in set $A$. Let $A^{(K)} = \{B \subseteq A \mid |B| = K\}$ be the set of $K$-subsets of $A$.
Let $\mathcal{U} = [n_u]$ denote the set of users and $\mathcal{I} = [n]$ denote the set of items, where $n_u$ and $n$ are the numbers of users and items, respectively. Without loss of generality, we assume the users and the items are numbered as $1, \dots, n_u$ and $1, \dots, n$, respectively.

\subsection{Package-to-group recommendations}

Let $g$ be the number of members in a group, and let $\mathcal{G} \in \mathcal{U}^{(g)}$ be a group, i.e., a set of users. Without loss of generality, we assume $\mathcal{G} = [g]$. We also assume that the preference $\boldR_{ui} \in \mathbb{R}_{\ge 0}$ (e.g., rating) of user $u \in \mathcal{U}$ for item $i \in \mathcal{I}$ for each user and item is known, which is a common assumption made in previous studies \cite{qi2016recommending, serbos2017fairness, kaya2020ensuring}. When a rating matrix at hand is incomplete, we complete the matrix using off-the-shelf completion algorithms. For example, \citet{serbos2017fairness} utilized the item-to-item collaborative filtering \cite{sarwar2001item}, and \citet{kaya2020ensuring} employed matrix factorization \cite{pilaszy2010fast} to complete rating matrices.

A package is represented by a set of items. Let $K \in \mathbb{Z}_+$ be the size of a package. We recommend a list of packages to group $\mathcal{G}$. A straightforward method is to recommend the package $\mathcal{P}$ with the maximum total preference: $\argmax_{\mathcal{P} \in \mathcal{I}^{(K)}} S(\mathcal{P}) \vcentcolon= \sum_{i \in \mathcal{P}} \sum_{u \in \mathcal{G}} \boldR_{ui}$. However, this is problematic because the optimal package may be unfair to the minority people in the group, as discussed in the introduction. Several fairness measures have been proposed to avoid unfair recommendations. In this paper, we consider the following two measures:

\begin{itemize}
    \item \textbf{Proportionality} \cite{qi2016recommending}: We say user $u$ likes item $i$ if $\boldR_{ui}$ is ranked in the top-$\Delta$\% in $\boldR_{u\colon}$. We say that package $\mathcal{P}$ is $m$-proportional for user $u$ if user $u$ likes at least $m$ items in package $\mathcal{P}$. The $m$-proportionality of package $\mathcal{P}$ is the number of members for whom package $\mathcal{P}$ is $m$-proportional.
    \item \textbf{Envy-freeness} \cite{serbos2017fairness}: We say user $u$ is envy-free for item $i$ if $\boldR_{ui}$ is ranked in the top-$\Delta$\% in $\{\boldR_{vi} \mid v \in \mathcal{G}\}$. We say that package $\mathcal{P}$ is $m$-envy-free for user $u$ if user $u$ is envy-free for at least $m$ items in package $\mathcal{P}$. The $m$-envy-freeness of package $\mathcal{P}$ is the number of members for whom package $\mathcal{P}$ is $m$-envy-free.
\end{itemize}

These measures have been successfully used in the fairness of package-to-group recommendations \cite{qi2016recommending, serbos2017fairness}.
In the following, we assume that $m = 1$ (i.e., single proportionality and single envy-freeness). We describe how to extend our method to multi proportionality and envy-freeness in Section \ref{sec: extend}.

Previous works proposed methods for optimizing the above-mentioned fairness measures. However, they recommend only one package, which is in contrast to the standard setting in which several options are suggested to users. In this paper, we solve the more general problem of enumerating all fair packages. Specifically, given threshold $\tau \in \mathbb{Z}_+$, we enumerate all packages whose fairness values are larger than $\tau$. Formally, we solve the following problem.

\vspace{0.05in}
\noindent \textbf{Problem 1. Enumerating Fair Packages.} \\
\textbf{Given:} Group $\mathcal{G}$, package size $K$, rating matrix $\boldR \in \mathbb{R}_{\ge 0}^{n_u \times n}$, and threshold $\tau \in \mathbb{N}_+$. \\
\textbf{Output:} All packages with proportionality at least $\tau$: $\{\mathcal{P} \in \mathcal{I}^{(K)} \mid \text{proportionality(} \mathcal{P} \text{)} \ge \tau\}$, and all packages with envy-freeness at least $\tau$: $\{\mathcal{P} \in \mathcal{I}^{(K)} \mid \text{envy-freeness(} \mathcal{P} \text{)} \ge \tau\}$.

\vspace{0.05in}
We emphasize that we do not resort to any approximations. The output should contain \emph{all} packages that meet the constraints and should not contain any other packages. The size of the output may be exponentially large. Thus, we output packages in compressed form, which is subsequently introduced in Section \ref{sec: dd}.

It should be noted that we do not necessarily use all output packages. Rather, for example, we enumerate (1) the set $A$ of packages with high proportionality, (2) the set $B$ of packages with high envy-freeness, and (3) the set $C$ of packages with high preferences, and find an ideal package by inspecting common packages contained in all of the three sets. Although the number of the filtered packages in $A \cap B \cap C$ may be small, each of $A$, $B$, and $C$ is extremely large; thus, we require an efficient enumeration algorithm.

\subsection{Fixed Parameter Tractable (FPT)}

We review FPT algorithms briefly. Let $N$ denote the size of an input, and let $k$ denote a parameter of a problem. For example, in this paper, we consider the size of a group as the parameter. Recall that a problem in P (complexity class) is efficiently solved only if $N$ is moderately small. Therefore, if the fair package-to-group recommendation problem is in P, it is efficiently solved even if, e.g., the number of users is $10^6$, number of items is $10^6$, package size is $10^6$, \emph{and} group size is $10^6$, which is an improbable scenario. The main idea of an FPT algorithm is restricting a parameter to be small. The fair package-to-group recommendation problem with the group size parameter is FPT if the problem is efficiently solved when the group size is small, e.g., $10$, but it requires that the problem is efficiently solved even when the numbers of users and items and the package size are large, e.g., $10^6$. Therefore, an FPT algorithm efficiently solves large instances with many users and items only if the group size is small.

Formally, algorithms that run in $O(\text{Poly}(N) f(k))$ time for some polynomial $\text{Poly}(N)$ and some function $f(k)$ are called FPT algorithms with respect to parameter $k$. Here, although the degree of $\text{Poly}(N)$ should be independent of $k$, $f(k)$ can be exponentially large. A problem with parameter $k$ is FPT if it has an FPT algorithm. It should be noted that this is different from the definition of XP (complexity class)\footnote{Algorithms that run in $O(N^{f(k)})$ time are called XP algorithms, i.e., they run in polynomial time for each fixed $k$.}, which requires a polynomial-time algorithm when the parameter is a constant. For example, a problem is in XP when it is solvable in $O(N^k)$ time; however, it is not necessarily FPT because the degree of the polynomial depends on $k$. A problem is FPT when it is solvable in $O(N 2^k)$; however, it is not necessarily in P because it takes exponential time when $k$ is linearly large. To illustrate the difference of XP and FPT algorithms, suppose $N = 10^6$ and $k = 10$. An $O(N^k)$ time algorithm requires roughly $N^k = 10^{60}$ operations, and thus, it is not tractable. By contrast, an $O(N 2^k)$ algorithm requires roughly $N 2^k = 10^6 \cdot 2^{10} \approx 10^9$ operations, and thus, it is tractable. Moreover, there exists a theoretical gap between FPT and XP problems \cite{downey1995fixed}. FPT is a well established field, and textbooks and surveys \cite{downey2012parameterized, cygan2015parameterized} are available for more details on FPT algorithms.

In this paper, we prove that the fair package-to-group recommendation problem with the group size parameter is FPT. In the literature, many examples considering small group sizes (e.g., families) are reported\footnote{Notably, some works consider larger groups, e.g., $k \le 20$ in \cite{anagnostopoulos2017tour}}. For example, \citet{serbos2017fairness} set the default group size as $8$ and the maximum group size as $16$. \citet{Lin2017fairness} and \citet{kaya2020ensuring} set the maximum group size as $8$. Therefore, our FPT algorithm efficiently solves their cases with a theoretical guarantee.

\subsection{Zero-suppressed Decision Diagrams (ZDDs)} \label{sec: dd}

A ZDD \cite{minato1993zero} is an efficient data structure for representing a family of sets. The universe set of set families is the item set $\mathcal{I}$ throughout the paper, and the order of variables are the same as the indices of the item set. A ZDD is represented by a rooted directed acyclic graph, and it contains a top node $\top$, a bottom node $\bot$, and branch nodes. Each branch node $v$ is labeled with an index $i_v \in \mathcal{I}$ of an item and has two outgoing edges: $0$-edge and $1$-edge. The label of the head node should be larger than that of the tail node (i.e., ordered). Figure \ref{fig: example} (a) shows an example of ZDDs. A ZDD represents a family of sets. Let $\mathcal{X} = \{i_1, i_2, \dots, i_l\} \subseteq \mathcal{I}$ be a set of items such that $i_1 < i_2 < \dots < i_l$. $\mathcal{X}$ is contained in a ZDD if and only if there exists a path $(. \xrightarrow{0})^* ~i_1 \xrightarrow{1} (. \xrightarrow{0})^* ~i_2 \xrightarrow{1} \cdots (. \xrightarrow{0})^* ~i_l \xrightarrow{1} (. \xrightarrow{0})^* \top$ that starts at the root node, where a dot denotes an arbitrary node, $\xrightarrow{0}$ denotes a $0$-edge, and $\xrightarrow{1}$ denotes a $1$-edge, and an asterisk indicates zero or more occurrences of the preceding element. For example, the red bold path in Figure \ref{fig: example} (b) represents $\{1, 3, 5\}$, and the red bold path in Figure \ref{fig: example} (c) corresponds to $\{2, 3, 4\}$. Therefore, the ZDD in Figure \ref{fig: example} (a) represents a set family $\{\{1, 3, 5\}, \{1, 2\}, \{2, 3, 4\}, \{4, 5\}\}$.

We use the Fraktur fonts such as $\mathfrak{A}$, $\mathfrak{B}$, and $\mathfrak{C}$ to denote ZDDs. We use the same variable for the set family represented by a ZDD. We consider the size of a ZDD to be the number of branch nodes in the ZDD. For a set family $\mathfrak{A}$, let $|\mathfrak{A}|$ denote the number of sets in $\mathfrak{A}$, and let $\|\mathfrak{A}\|$ denote the number of branch nodes in ZDD $\mathfrak{A}$.

The most basic operation for ZDDs is reduction, which is depicted in Figure \ref{fig: reduction}. Specifically, duplicate nodes that point to the same children and have the same item label are merged, and a node whose $1$-edge is connected to the bottom node is skipped. This process is continued until no nodes are redundant\footnote{We say a node is redundant if its $1$-edge is connected to the bottom node, or some other nodes with the same item label point to the same children.}. ZDD $\mathfrak{A}$ can be reduced in $O(\|\mathfrak{A}\|)$ time. The reduction does not change the set family that $\mathfrak{A}$ represents and decreases the size of the ZDD. It is known that a set family has a unique reduced ZDD, and any ZDD can reach the reduced form by any sequence of reduction operations. In some literature, ZDDs stand only for reduced ZDDs, whereas we use both non-reduced and reduced ZDDs in this paper.

\begin{figure}
    \centering
    \includegraphics[width=\hsize]{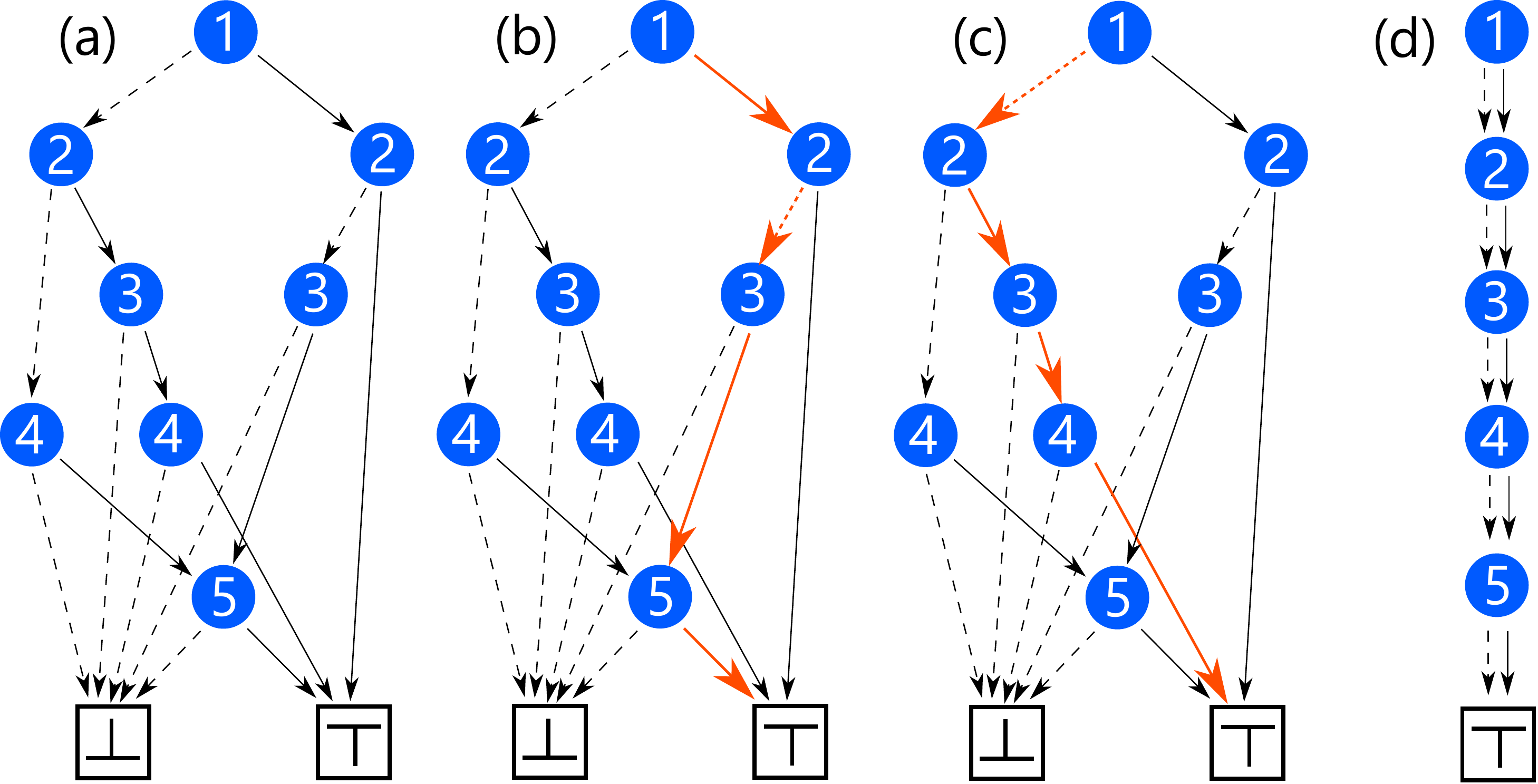}
    \caption{Examples of ZDDs: dashed and solid lines represent $0$- and $1$-edges, respectively. (a) ZDD representing family $\{\{1, 3, 5\}, \{1, 2\}, \{2, 3, 4\}, \{4, 5\}\}$. (b) Path corresponding to $\{1, 3, 5\}$. (c) Path corresponding to $\{2, 3, 4\}$. (d) ZDD representing the power set of $\{1, 2, 3, 4, 5\}$.}
    \label{fig: example}
    \vspace{-0.1in}
\end{figure}

\begin{figure}
    \centering
    \includegraphics[width=\hsize]{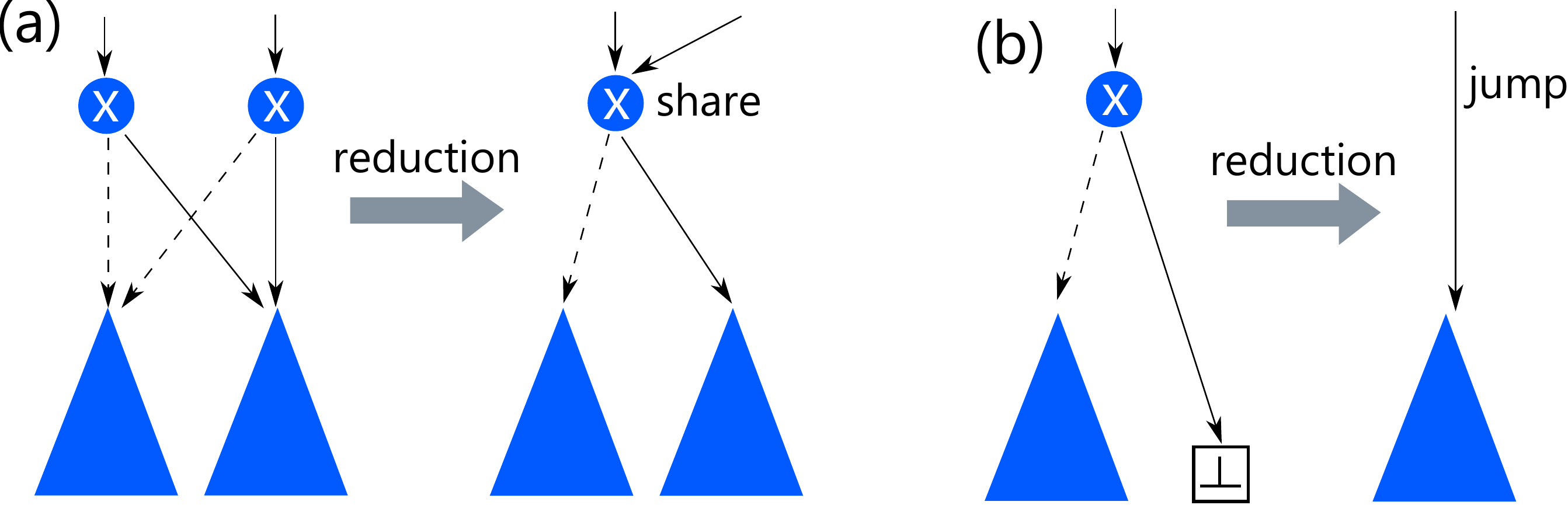}
    \caption{Reduction of ZDDs: (a) If two nodes of the same index point to the same successors, they merge into a single node. (b) If the $1$-edge of a node points to the bottom node, that node is skipped. Intuitively, if no sets contain item $x$ (i.e., $1$-edge points to the bottom node), item $x$ can be ignored.}
    \label{fig: reduction}
    \vspace{-0.1in}
\end{figure}

One of the benefits of ZDDs is that ZDDs consume less memory than raw representations. For example, Figure \ref{fig: example} (d) shows a ZDD that represents the power set of $\{1, 2, 3, 4, 5\}$. Although this set family contains $2^5 = 32$ sets, the ZDD contains only five nodes and ten edges. In general, the advantage tends to increase exponentially as the size of the universe set increases. Besides, the size of a reduced ZDD is at most the total number of elements in the sets in the family. Specifically, $\|\mathfrak{A}\| \le \sum_{\mathcal{A} \in \mathfrak{A}} |\mathcal{A}|$ holds because when a new set $\mathcal{X}$ is inserted to a ZDD, at most $|\mathcal{X}|$ nodes are created. Therefore, a ZDD is as memory efficient as a raw description, even in the worst case, and performs much better in practice. 

Another benefit is that ZDDs support many operations for set families \cite{knuth2009art}. The celebrated apply operation \cite{bryant1986graph, minato1993zero} is an example, which efficiently merges two ZDDs. Specifically, it computes the reduced ZDDs that represent $\mathfrak{A} \cup \mathfrak{B}$, $\mathfrak{A} \cap \mathfrak{B}$, and $\mathfrak{A} \backslash \mathfrak{B}$ from ZDDs $\mathfrak{A}$ and $\mathfrak{B}$ in $O(\|\mathfrak{A}\| \|\mathfrak{B}\|)$ time in the worst case. It runs in linear time in practice without pathological cases \cite{minato1993zero, bryant1986graph}. The crux of the apply operation is that it manipulates compressed ZDDs directly, and the running time is independent of the number of sets contained in the ZDDs but only dependent on the size of the ZDDs. Recall that the number of sets can be exponentially larger than the size of the ZDD (e.g., the power set). Thus, the benefit can be exponentially large.

\section{Method}

We introduce our proposed method and analyze time complexity. Although our proposed method is simple and easy to implement, it has a strong theoretical property on the fixed parameter tractability.

\subsection{Maximization Problem} \label{sec: max}

As we highlighted in the introduction, the enumeration problem is more difficult than the maximization problem. We first show that the maximization problem can be efficiently solved, even though it is an NP-hard problem.
We consider the proportionality here. Formally, the maximization problem is defined as follows:

\vspace{0.05in}
\noindent \textbf{Problem 2. Maximizing Fairness.} \\
\textbf{Given:} Group $\mathcal{G}$, package size $K$, and rating matrix $\boldR \in \mathbb{R}_{\ge 0}^{n_u \times n}$. \\
\textbf{Output:} The maximum proportionality of all possible packages, i.e., $\max \{\text{proportionality} (\mathcal{P}) \mid \mathcal{P} \in \mathcal{I}^{(K)}\}$.

\vspace{0.05in}
A naive algorithm for this problem is to check all $K$-sets of items. However, this algorithm requires at least $O(n^K)$ time and is not FPT. The dynamic programming algorithm introduced in this section efficiently solves this problem and serves as a foundation for the analysis of our proposed method in the next section.

\begin{theorem} \label{thm: max}
The maximizing fairness problem can be solved in $O(n (K + g) 2^g)$ time and is FPT with respect to the group size parameter.
\end{theorem}

\begin{proof}
For item $i \in \mathcal{I}$, integer $k \in \{0\} \cup [K]$, and subset $\mathcal{H} \subset \mathcal{G}$ of group $\mathcal{G}$, let $\boldB_{i, k, \mathcal{H}}$ be a Boolean variable that indicates whether there exists a package $\mathcal{P} \subset [i]^{(k)}$ that is liked by all users in $\mathcal{H}$. Let $\mathcal{S}_i \subset \mathcal{G}$ be the set of users that like item $i$. We show that dynamic programming efficiently computes the values of $\boldB$ in a recursive manner. Specifically, the following recursive relations hold:
\begin{align*}
    \boldB_{i, 0, \emptyset} &= \textbf{True} \quad (0 \le i \le n), \\
    \boldB_{i, 0, \mathcal{H}} &= \textbf{False} \quad (0 \le i \le n, \mathcal{H} \neq \emptyset), \\
    \boldB_{0, k, \mathcal{H}} &= \textbf{False} \quad (1 \le k \le K), \\ 
    \boldB_{i, k, \mathcal{H}} &= \boldB_{i-1, k, \mathcal{H}}+ \sum_{\mathcal{F}\colon \mathcal{F} \cup \mathcal{S}_i = \mathcal{H}} \boldB_{i-1, k-1, \mathcal{F}} \quad(1 \le i \le n, 1 \le k \le K),
\end{align*}
where $+$ and $\sum$ denote Boolean OR. The first three equations follow from the definition. The last equation holds because there exists a package $\mathcal{P} \subseteq [i]^{(k)}$ that is liked by all users in $\mathcal{H}$ if (1) there already exists a package $\mathcal{P} \subseteq [i-1]^{(k)}$ that is liked by all users in $\mathcal{H}$ or (2) there exists a package $\mathcal{P} \subseteq [i-1]^{(k-1)}$ that is liked by all users in $\mathcal{F}$ for some subset $\mathcal{F}$, and the remaining users $\mathcal{H} \backslash \mathcal{F}$ like item $i$. In the latter case, the size of the package increases by one, thereby increasing index $k$. Although the summation in the fourth equation may be exponentially large, only one term is included on average because only one $\mathcal{H}$ satisfies $\mathcal{F} \cup \mathcal{S}_i = \mathcal{H}$ for each $\mathcal{F}$. Algorithm \ref{algo: max} shows the pseudo code. For a fixed $\mathcal{H}$, it is cumbersome to iterate $\mathcal{F}$ that satisfies $\mathcal{F} \cup \mathcal{S}_i = \mathcal{H}$ as in the formula above. To handle it, we iterate $\mathcal{F}$ instead of $\mathcal{H}$ in the code, thereby, it becomes clear where to add $\boldB_{i-1, k, \mathcal{F}}$, i.e., to $\boldB_{i, k+1, \mathcal{F} \cup \mathcal{S}_i}$. In Line 3, the variables are initialized as described above. In Lines 4--11, the values of $\boldB$ are computed in a bottom-up manner. In Line 12, the maximum satisfied set is computed. The computational bottleneck of this algorithm is the triple loop in Lines 4--11. If $\mathcal{H}$ and $\mathcal{S}_i$ are represented by Boolean arrays, $\mathcal{H} \cup \mathcal{S}_{i+1}$ can be computed in $O(g)$ time. Hence, the time complexity of this algorithm is $O(n(K + g)2^g)$.
\end{proof}

It should be noted that previous works proved the problem is NP-hard \cite{serbos2017fairness}, concluded that it was not efficiently solvable, and resorted to approximations. However, the above theorem shows the problem can be solved efficiently and \emph{exactly} when the group size is small, which is the case in previous works \cite{serbos2017fairness, Lin2017fairness, kaya2020ensuring}. 

Maximization of the envy-freeness is also solvable in $O(n (K + g) 2^g)$ time by replacing $\mathcal{S}_i$ with the set of envy-free users for $i$.

\setlength{\textfloatsep}{5pt}
\begin{algorithm2e}[t]
\caption{Maximizing Fairness}
\label{algo: max}
\DontPrintSemicolon 
\nl\KwData{Group $\mathcal{G}$, Package size $K$, Rating matrix $\boldR \in \mathbb{R}_{\ge 0}^{n_u \times n}$.}
\nl\KwResult{Maximum proportionality.}
    \nl $\boldB_{i, k, \mathcal{H}} \leftarrow \begin{cases}
        \textbf{True} & (k = 0, \mathcal{H} = \emptyset) \\
        \textbf{False} & (\text{otherwise})
    \end{cases}$ \;
    \nl \For{$i\gets1$ \KwTo $n$}{
    \nl     $\mathcal{S}_{i} \leftarrow \{ u \in \mathcal{G} \mid \boldR_{ui} \text{ is ranked in the top-}\Delta \text{ \% in }\boldR_{u\colon}\}$ \;
    \nl     \For{$\mathcal{F} \subseteq \mathcal{G}$}{
    \nl         $\mathcal{H}\leftarrow \mathcal{F} \cup \mathcal{S}_{i}$  \tcp*{Add item $i$}
    \nl         \For{$k\gets0$ \KwTo $K$}{
    \nl             \If{$k \le K - 1$}{
    \nl                 $\boldB_{i, k+1, \mathcal{H}} \leftarrow \boldB_{i, k+1, \mathcal{H}} + \boldB_{i-1, k, \mathcal{F}}$ \;
                    }
    \nl             $\boldB_{i, k, \mathcal{F}} \leftarrow \boldB_{i, k, \mathcal{F}} + \boldB_{i-1, k, \mathcal{F}}$ \;
                }
            }
        }
    \nl \textbf{return} $\max \{|\mathcal{H}| \mid \mathcal{H} \subset \mathcal{G}, \boldB_{n, K, \mathcal{H}} = \textbf{True}\}$ \;
\end{algorithm2e}

\subsection{Fair Package Enumeration}

This section describes our proposed method \textsc{FAPE} (\underbar{FA}ir \underbar{P}ackage \underbar{E}numeration). We consider the proportionality first. Thus, the ZDD we construct in the following represents all packages with proportionality at least $\tau$. i.e., $\{\mathcal{P} \in \mathcal{I}^{(K)} \mid \text{proportionality(} \mathcal{P} \text{)} \ge \tau\}$. 

The ZDD we consider in this section has the same structure as the DP table $\boldB$ used in the proof of Theorem \ref{thm: max}. Specifically, let the vertex set be $\mathcal{V} = \{\mathfrak{x}_{i, k, \mathcal{H}} \mid i \in \mathcal{I}, k \in \{0, 1, \cdots, K\}, \mathcal{H} \subseteq \mathcal{G} \} \cup \{\bot, \top\}$. The label of node $\mathfrak{x}_{i, k, \mathcal{H}}$ is item $i$. The root node is $\mathfrak{x}_{1, 0, \emptyset}$. Intuitively, $\mathfrak{x}_{i, k, \mathcal{H}}$ represents a set of packages that contain $k$ items in $[i]$ and are liked by $\mathcal{H}$. We assume that the all undefined nodes (e.g., $\mathfrak{x}_{0, K+1, \emptyset}$) denote the bottom node. We construct the ZDD as follows. \textbf{Internal states:} Recall that $\mathcal{S}_i \subseteq \mathcal{G}$ is the set of users that like item $i$. For each $i \in [n-1]$, the $0$-edge of node $\mathfrak{x}_{i, k, \mathcal{H}}$ points to $\mathfrak{x}_{i+1, k, \mathcal{H}}$ because the set of satisfied members does not change if an item is not inserted. The $1$-edge points to $\mathfrak{x}_{i, k+1, \mathcal{S}_i \cup \mathcal{H}}$ because members in $\mathcal{S}_i$ are newly satisfied by item $i$ and the size of the package increases by $1$. \textbf{Last states:} The $0$-edge of node $\mathfrak{x}_{n, k, \mathcal{H}}$ points to $\top$ if $|\mathcal{H}| \ge \tau$ and $k = K$ and points to $\bot$ otherwise by the definition of the solutions. Recall that package $\mathcal{P}$ is included (i.e., connected to $\top$) if the size $|\mathcal{P}| = k$ is $K$ and the number of satisfied members $\mathcal{H}$ is at least $\tau$. In the $1$-edge case, the sizes of the packages are incremented and $\mathcal{S}_n$ is newly satisfied. Thus, the $1$-edge of node $\mathfrak{x}_{n, k, \mathcal{H}}$ points to $\top$ if $|\mathcal{H} \cup \mathcal{S}_n| \ge \tau$ and $k = K-1$ and points to $\bot$ otherwise.

As in case of Theorem \ref{thm: max}, for $i \in [n]$ and $\mathcal{P} \subseteq [i]$, package $\mathcal{P}$ contains $k$ items and is liked by all users in $\mathcal{H}$ if the corresponding path reaches $\mathfrak{x}_{i+1, k, \mathcal{H}}$. Therefore, a path from the root reaches the top node if and only if the corresponding package is liked by at least $\tau$ users and contains exactly $K$ items. The size of this ZDD is $n (K+1) 2^g$ by construction.

The ZDD described above is not reduced. One way to obtain the reduced ZDD is to construct the ZDD described above and then reduce it. However, this two-step method may produce redundant nodes in the intermediate step and is memory inefficient. By contrast, \textsc{FAPE} constructs the reduced ZDD directly all at once, from bottom to top, avoiding duplicate (Figure \ref{fig: reduction} (a)) and redundant (Figure \ref{fig: reduction} (b)) nodes.

\setlength{\textfloatsep}{5pt}
\begin{algorithm2e}[t]
\caption{\textsc{FAPE}}
\label{algo: enum}
\DontPrintSemicolon 
\nl\KwData{Group $\mathcal{G}$, Package size $K$, Rating matrix $\boldR \in \mathbb{R}_{\ge 0}^{n_u \times n}$, threshold $\tau \in \mathbb{Z}_+$.}
\nl\KwResult{All packages with proportionality at least $\tau$.}
    \nl \For{$i \gets n$ \KwTo $1$}{
    \nl     $\mathcal{S}_i \leftarrow \{ u \in \mathcal{G} \mid \boldR_{ui} \text{ is ranked in top-}\Delta \text{ \% in }\boldR_{u\colon}\}$ \;
    \nl     \For{$\mathcal{H} \subseteq \mathcal{G}$}{
    \nl         $\mathcal{H}' \leftarrow \mathcal{H} \cup \mathcal{S}_{i}$  \tcp*{Add item $i$}
    \nl         \For{$k\gets0$ \KwTo $K$}{
    \nl             \uIf{$i = n$}{
    \nl                 hi $\leftarrow \top \textbf{ if } |\mathcal{H}'| \ge \tau \textbf{ and } k = K-1 \textbf{ else } \bot$ \;
    \nl                 low $\leftarrow \top \textbf{ if } |\mathcal{H}| \ge \tau \textbf{ and } k = K \textbf{ else } \bot$ \;
                    }
    \nl             \uElseIf{$k = K$}{
    \nl                 hi $\leftarrow \bot;$ ~low $\leftarrow \mathfrak{s}_{i+1, k, \mathcal{H}}$ \;
                    }
    \nl             \Else{
    \nl                 hi $\leftarrow \mathfrak{s}_{i+1, k+1, \mathcal{H}'};$ ~low $\leftarrow \mathfrak{s}_{i+1, k, \mathcal{H}}$ \;
                    }
    \nl             \uIf{\textup{hi } $= \bot$}{
    \nl                 $\mathfrak{s}_{i, k, \mathcal{H}} \leftarrow$ low \tcp*{Skip (Fig. \ref{fig: reduction} (b))}
                    }
    \nl             \uElseIf{$\mathfrak{p}_{i, \textup{hi}, \textup{low}} \textup{ \textbf{is not None}}$}{
    \nl                 \tcp{Duplicate node found.}
    \nl                 $\mathfrak{s}_{i, k, \mathcal{H}} \leftarrow \mathfrak{p}_{i, \textup{hi}, \textup{low}}$ \tcp*{Merge (Fig. \ref{fig: reduction} (a))}
                    }
    \nl             \Else{
    \nl                 \tcp{No duplicate nodes found.}
    \nl                 $\mathfrak{p}_{i, \textup{hi}, \textup{low}} \leftarrow \mathfrak{x}_{i, k, \mathcal{H}}; ~\mathfrak{s}_{i, k, \mathcal{H}} \leftarrow \mathfrak{x}_{i, k, \mathcal{H}}$ \;
    \nl                 $\mathfrak{x}_{i, k, \mathcal{H}}.\text{edges} \leftarrow (\text{hi}, \text{low})$ \;
                    }
                }
            }
        }
    \nl \textbf{return} $\mathfrak{s}_{1, 0, \emptyset}$ \;
\end{algorithm2e}

Algorithm \ref{algo: enum} shows the pseudo-code. hi and low stand for $1$- and $0$-edges, respectively. Nodes $\mathfrak{x}_{i, k, \mathcal{H}}$ are divided into equivalence classes induced by the reduction operations (Figure \ref{fig: reduction}). Merged nodes belong to the same equivalence class. A skipped node is identified with the node pointed by the $0$-edge. $\mathfrak{s}_{i, k, \mathcal{H}}$ is the representative node of the equivalence class to which node $\mathfrak{x}_{i, k, \mathcal{H}}$ belongs. Recall that when two states $\mathfrak{x}_{i, k, \mathcal{H}}$ and $\mathfrak{x}_{i, k', \mathcal{H'}}$, for some $k, k', \mathcal{H}, \mathcal{H}'$, point to the same successors (hi, low), they are identified by the reduction operation (Figure \ref{fig: reduction} (a)). $\mathfrak{p}_{i, \text{hi}, \text{low}}$ is the representative node that points to (hi, low) with index $i$. We use the first state that points to (hi, low) as the representative node $\mathfrak{p}_{i, \text{hi}, \text{low}}$ and identify the following states that point to the same successors with the first one $\mathfrak{p}_{i, \text{hi}, \text{low}}$.

The outermost loop in Line 3 iterates item $i$ in the descending order, i.e., \textsc{FAPE} constructs the reduced ZDD from bottom to top. \textbf{Last states:} In Lines 8--10, we determine the successors of the last variable nodes. If the resulting packages contain exactly $K$ items and are liked by at least $\tau$ users, the link points to the top node, and otherwise points to the bottom node. See also the corresponding descriptions in the redundant ZDDs (``Last states'' on the left column). \textbf{Border states:} In Lines 11--12, we determine the successors of the nodes that contain $K$ items already. If a new item is inserted into such a package (i.e., $1$-edge), the resulting package contains more than $K$ items, and thus, should point to the bottom node. The $0$-edges are determined similarly as the internal states. \textbf{Internal states:} In Lines 13--14, we determine the successors of other nodes. See also the corresponding descriptions in the redundant ZDDs (``Internal states'' on the previous page). \textbf{Skipping:} In Lines 15--16, if the $1$-edge points to the bottom node, the node is skipped and identified with the node pointed by the $0$-edge. \textbf{Merging:} In Lines 17--19, if there already exists a node that points to the same successors, we merge the current node to the existing node. Otherwise, a new node is created in Lines 21--23. We use a hashmap to implement $\mathfrak{p}$. A bucket sort-based algorithm \cite{sieling1993reduction, knuth2009art} can also be employed.

Clearly, the construed ZDD is reduced. There are no nodes whose $1$-edges point to the bottom node owing to Line 15. There are no two nodes that point to the same successor with the same index because Line 17 does not allow creating such nodes. In addition, clearly, the resulting ZDD represents all packages with proportionality at least $\tau$ and contains exactly $K$ items owing to the same reason that the (non-reduced) ZDD introduced at the beginning of this subsection. The time complexity of \textsc{FAPE} is $O(n (K + g) 2^g)$ because the outermost loop iterates $n$ times, the middle loop in Line 5 iterates $2^g$ times, the innermost loop in Line 7 iterates $K + 1$ times, and the union operation in Line 6 takes $O(g)$ time. Therefore, the following theorem holds.

\begin{theorem} \label{thm: enum}
The enumerating fair packages problem can be solved in $O(n (K + g) 2^g)$ time.
\end{theorem}

\begin{corollary}
The enumerating fair packages problem is FPT with respect to the group size parameter.
\end{corollary}

\begin{corollary}
The size of the reduced ZDD that represents all packages with proportionality at least $\tau$ is $O(n K 2^g)$.
\end{corollary}

Note that the size of the ZDD is bounded by $O(n K 2^g)$, even in the worst case. It is much smaller in practice than the bound due to the pruning stpdf in Lines 15--19 in Algorithm \ref{algo: enum}. Therefore, the filtering operations we will introduce in the next section can be efficiently conducted once the ZDD is constructed.

Enumerating envy-free packages is also solvable in $O(n (K + g) 2^g)$ time by replacing $\mathcal{S}_i$ with the set of users that are envy-free for $i$.

\subsection{Operations} \label{sec: op}

Enumerating all candidates is helpful when the number of candidates is small. However, it is impossible to investigate all candidates by hand when there are countless candidates. The crux of our proposed algorithm is that it can filter items by various operations in such cases. We introduce several examples of such operations.

\vspace{0.05in}
\sloppy \noindent \uline{\textbf{Union and Intersection.}} As we reviewed in Section \ref{sec: dd}, ZDDs support the apply operation, which computes the intersection and union of two ZDDs. Therefore, for any $\tau, \tau'$, we can enumerate $\mathfrak{I}_{\tau, \tau'} = \{ \mathcal{P} \subseteq \mathcal{I} \mid \text{proportionality}(\mathcal{P}) \ge \tau \textbf{ and } \text{envy-freeness}(\mathcal{P}) \ge \tau' \}$, by constructing two ZDDs by \textsc{FAPE} and conducting the apply operation. Similarly, $\{ \mathcal{P} \subseteq \mathcal{I} \mid \text{proportionality}(\mathcal{P}) \ge \tau \textbf{ or } \text{envy-freeness}(\mathcal{P}) \ge \tau' \}$ can be constructed. This is beneficial when either criterion is acceptable. Furthermore, for any $\tau$, we can compute the maximum envy-freeness of the packages with proportionality at least $\tau$ by investigating non-empty $\mathfrak{I}_{\tau, \tau'}$ with maximum $\tau'$. In other words, we can compute the Pareto optimal packages in the sense of having the best envy-freeness for a given proportionality.

We can impose additional constraints by constraint ZDDs and the intersection operation. A typical example is category constraints \cite{serbos2017fairness}. In music playlist recommendations, all items in a package should have the same category for consistency. The constraint ZDD can be constructed by (1) constructing ZDD $\mathfrak{C}_c$ that represents the power sets of items in each category $c$ (see Figure \ref{fig: example} (d)) and (2) merging them by the union operation. $\mathfrak{C}_c$ contains all packages that are composed only of items with category $c$. Thus, the merged ZDD contains all packages that are composed of a single category. In tour recommendation usage, if we set categories based on prefectures or regions, the POIs in a package are constrained to be geometrically close. In movie and tour package recommendations, all items in a package should have different categories to avoid monotony. The constraint ZDD can be constructed by dynamic programming similar to \textsc{FAPE} in this case.
In general, constraint ZDDs are domain-specific, and it may be time-consuming to build them. However, we can re-use them for all groups. Once we build the constraint ZDDs and store them in a database, we can enumerate the filtered packages by running \textsc{FAPE} and the intersection operation with the stored constraint ZDDs.

\vspace{0.05in}
\noindent \uline{\textbf{Fixing Items.}} There may be favorite or disliked items in a group. ZDDs support superset and exclusion operations. Specifically, given an item set $\mathcal{Q} \subseteq \mathcal{I}$ and a ZDD $\mathfrak{A}$, we can build a ZDD that represents supersets $\{ \mathcal{P} \mid \mathcal{Q} \subseteq \mathcal{P} \in \mathfrak{A}\}$ and exclusion $\{ \mathcal{P} \mid \mathcal{P} \in \mathfrak{A}, \mathcal{P} \cap \mathcal{Q} = \emptyset \}$ in $O(\|\mathfrak{A}\| + |\mathcal{Q}|)$ time.

\vspace{0.05in}
\noindent \uline{\textbf{Optimizing Scores.}} In general, fairness is not the only objective. Even if a package is envy-free, it is not helpful if all items are not liked by any users. A natural request is to obtain high total preference $S(\mathcal{P}) = \sum_{i \in \mathcal{P}} \sum_{u \in \mathcal{G}} \boldR_{ui}$. ZDDs can efficiently solve the linear Boolean programming problem \cite{knuth2009art}. Specifically, given a weight $w_i \in \mathbb{R}$ for each item $i \in \mathcal{I}$ and a ZDD $\mathfrak{A}$, we can obtain the maximum weight package $\argmax_{\mathcal{P} \in \mathfrak{A}} \sum_{i \in \mathcal{P}} w_i$ in $O(\|\mathfrak{A}\|)$ time. Top-$k$ packages can be also computed efficiently. Therefore, if we set $w_i = \sum_{u \in \mathcal{G}} \boldR_{ui}$, we can compute the package with the maximum total preference in a ZDD. If we set $w_i = -\text{Var}(\{\boldR_{ui} \mid u \in \mathcal{G})\})$, the resulting package contains items that all users rated similarly, which can be considered as another fairness criterion. We can apply these operations to the ZDDs constructed by \textsc{FAPE} and those filtered by the intersection operations mentioned above. 

\vspace{0.05in}
\noindent \uline{\textbf{Sampling Packages.}} Although optimization is a powerful tool, it outputs an extreme case in a ZDD. Some users may want to determine average packages to know the properties of the ZDD at hand. ZDDs support uniform and weighted sampling. Specifically, given ZDD $\mathfrak{A}$, uniform sampling outputs each package with probability $1/|\mathfrak{A}|$. Given a weight $w_i \in \mathbb{R}$ for each item $i \in \mathcal{I}$, weighted sampling outputs package $\mathcal{P} \in \mathfrak{A}$ with probability proportional to $\sum_{i \in \mathcal{P}} w_i$. These operations enable investigating ZDDs even if their sizes are large. Note that even uniform sampling may be satisfactory because the ZDDs constructed by \textsc{FAPE} contain only fair packages.

\subsection{Extensions} \label{sec: extend}

In this section, we describe how to extend our method to $m > 1$ briefly. For $i \in \mathcal{I}$, $k \in \{0, 1, \cdots, K\}$, $\mathcal{H} \in \{0, 1, \cdots, m\}^\mathcal{G}$, let $\mathfrak{x}_{i, k, \mathcal{H}}$ represent packages $\mathcal{P} \in \mathcal{I}^{(K)}$ such that user $u \in \mathcal{G}$ likes $\mathcal{H}_u$ items in $\mathcal{P}$. $\mathcal{H}_u = m$ means at least $m$ items. Then, the reduced ZDD can be constructed as \textsc{FAPE}, where $\mathcal{H}$ is not a Boolean array but a counter array. The time complexity is $O(n (K + g) (m + 1)^g)$. Therefore, the problem is still FTP if it is parameterized with both $g$ and $m$. We note that the ZDDs constructed with different $m$ can be combined by the union and intersection operations for finer-grained fairness.

\section{Experiments}

\begin{figure*}[tb]
\begin{minipage}{0.22\hsize}
\centering
\includegraphics[width=\hsize]{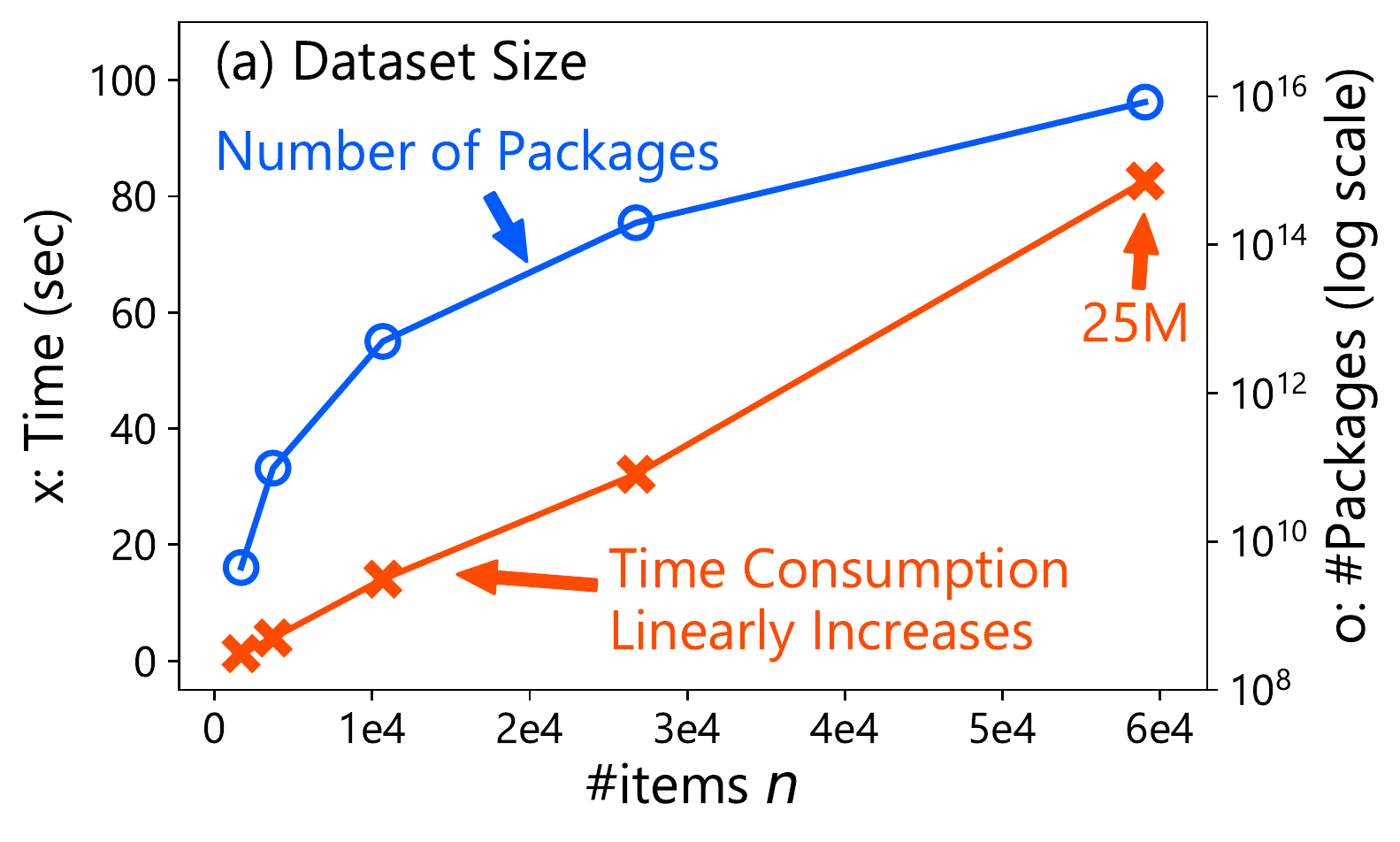}
\end{minipage}
\begin{minipage}{0.22\hsize}
\centering
\includegraphics[width=\hsize]{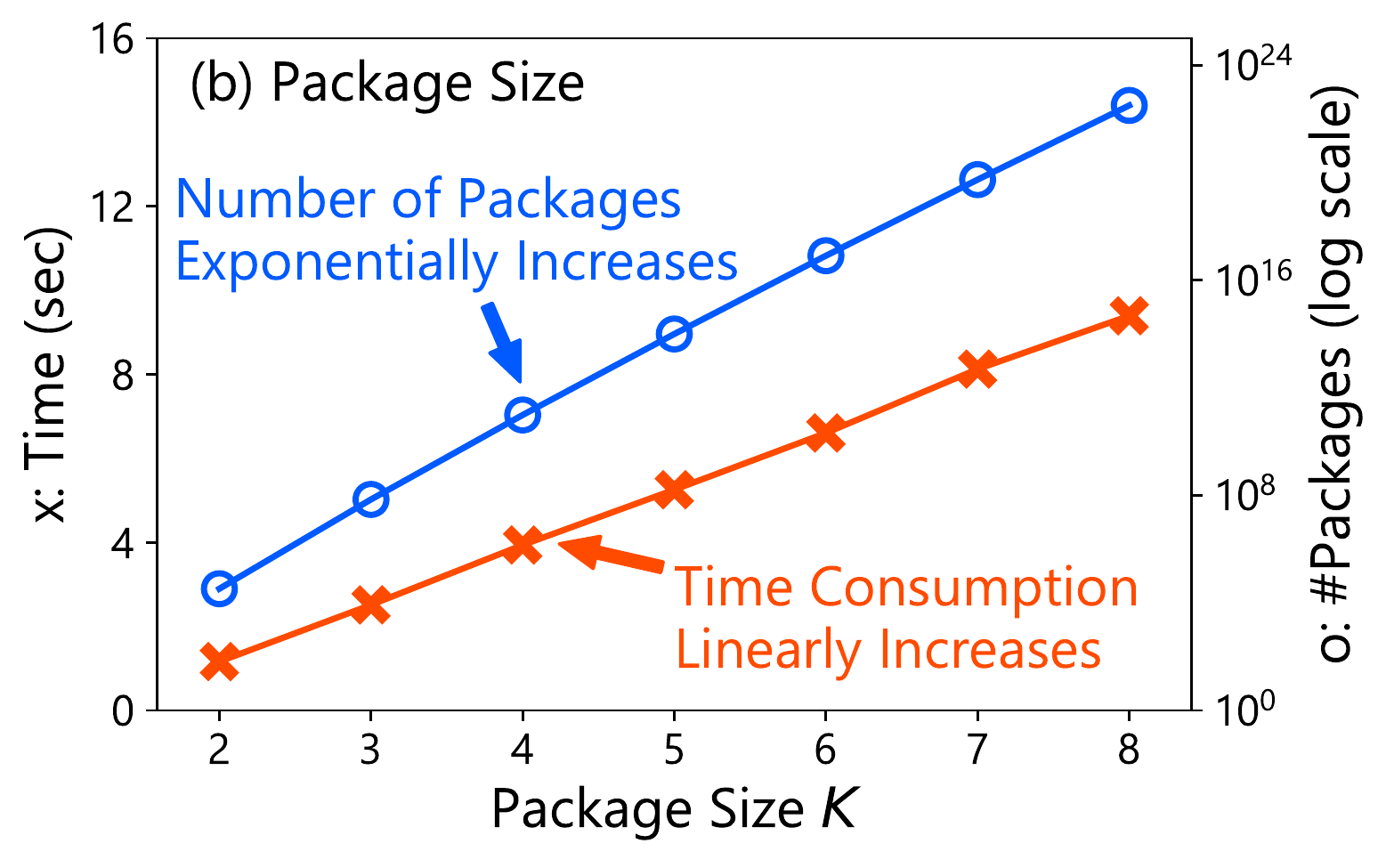}
\end{minipage}
\begin{minipage}{0.22\hsize}
\centering
\includegraphics[width=\hsize]{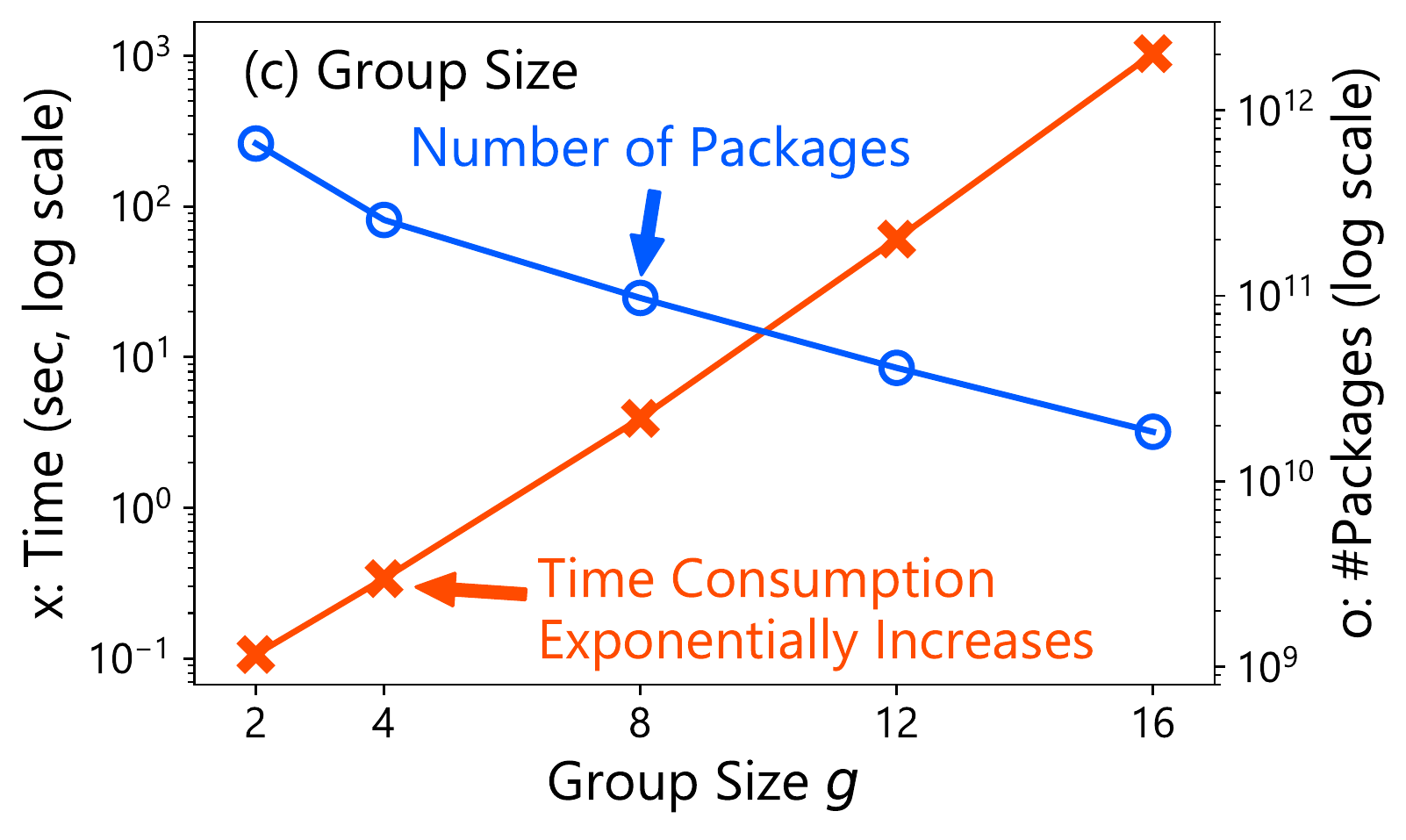}
\end{minipage}
\begin{minipage}{0.22\hsize}
\centering
\includegraphics[width=\hsize]{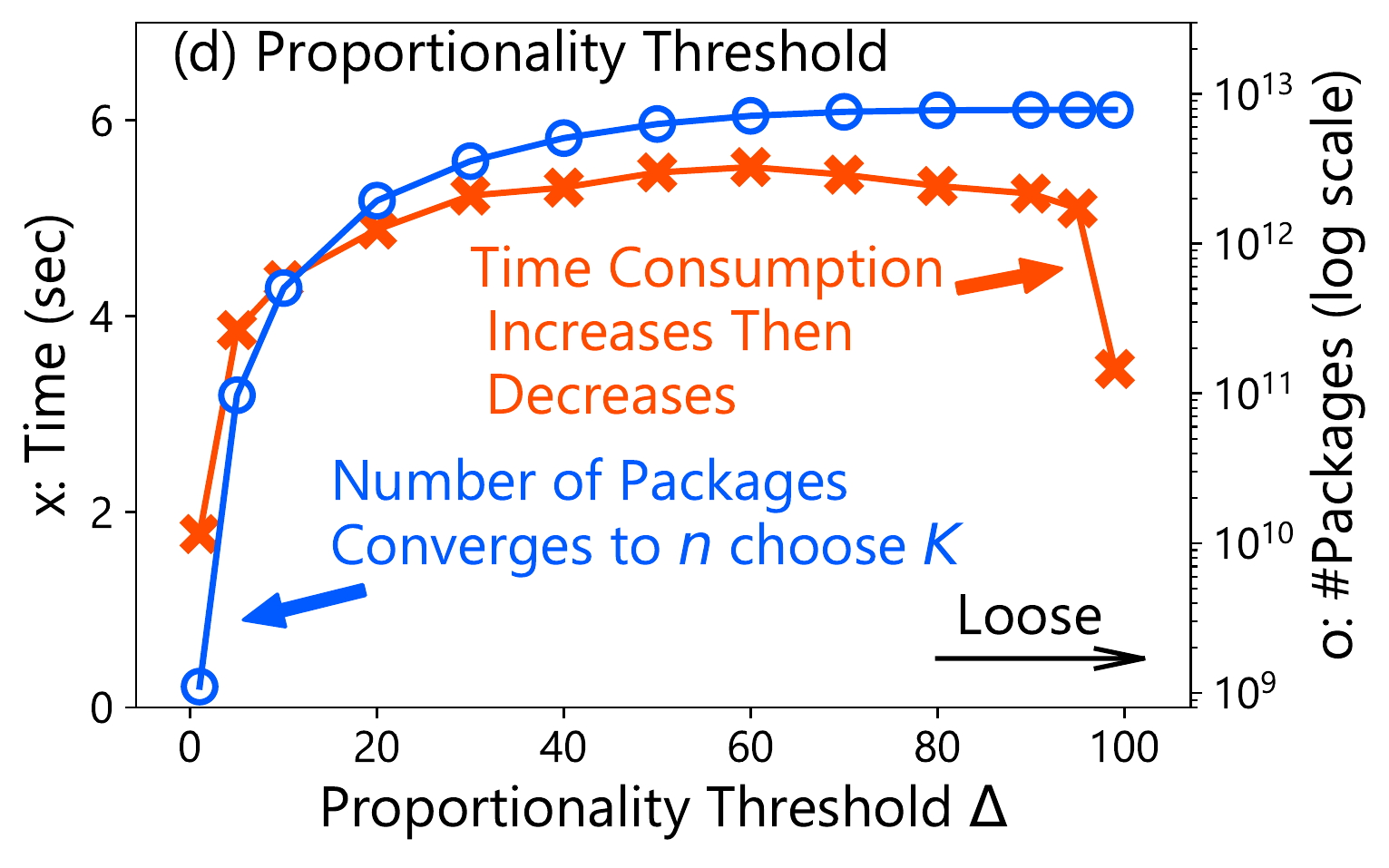}
\end{minipage}
\begin{minipage}{0.22\hsize}
\centering
\includegraphics[width=\hsize]{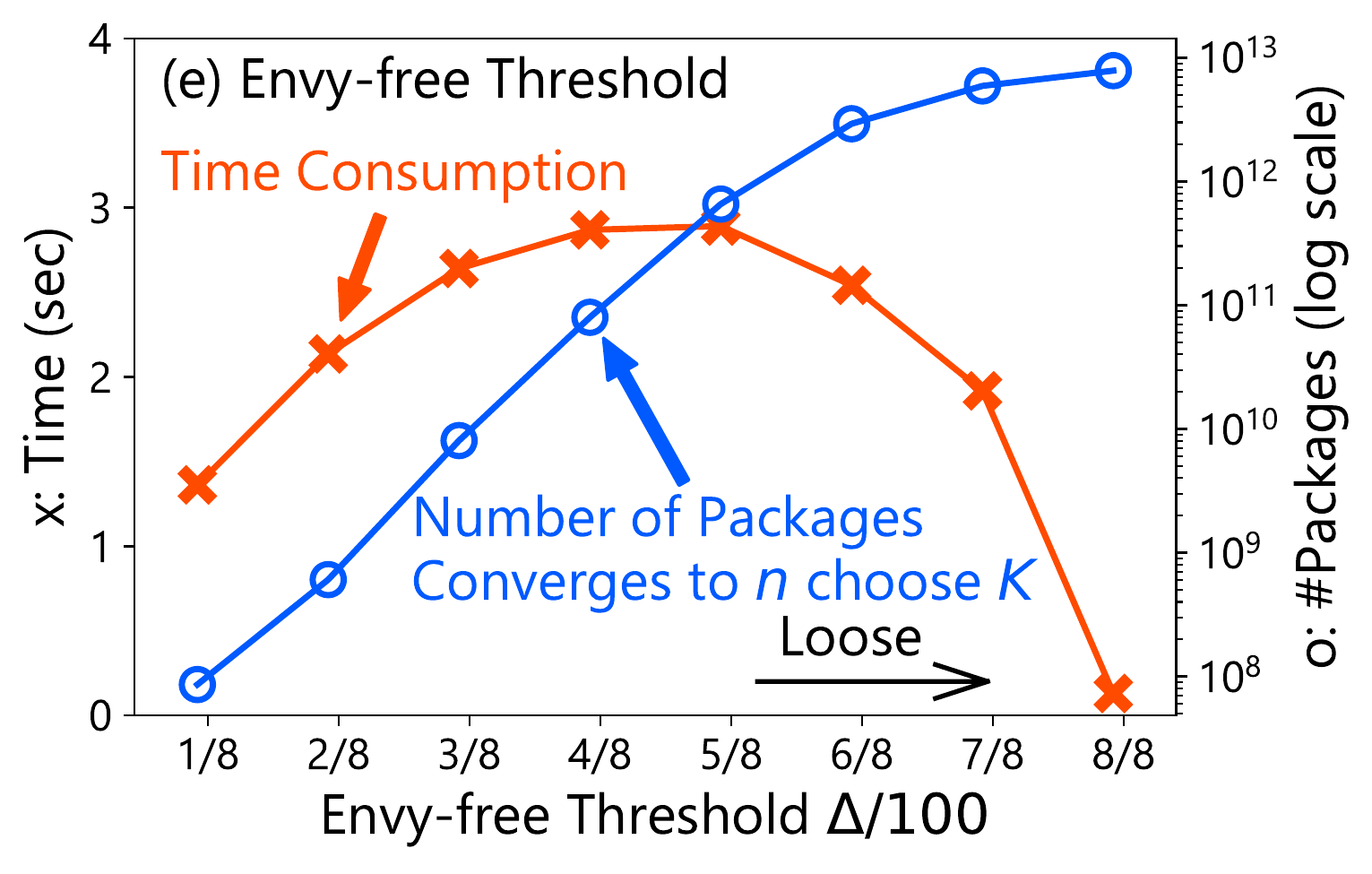}
\end{minipage}
\begin{minipage}{0.22\hsize}
\centering
\includegraphics[width=\hsize]{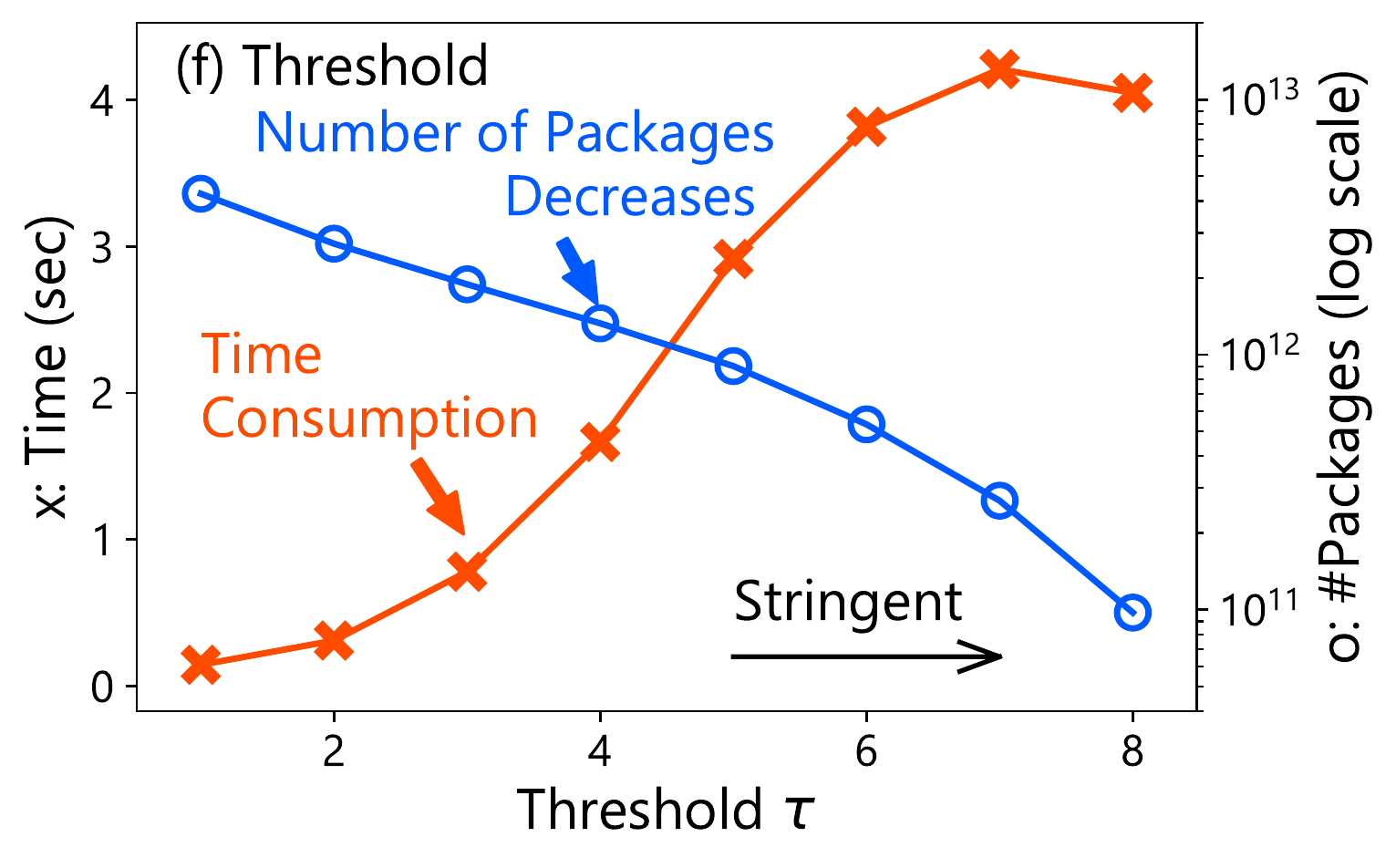}
\end{minipage}
\begin{minipage}{0.22\hsize}
\centering
\includegraphics[width=\hsize]{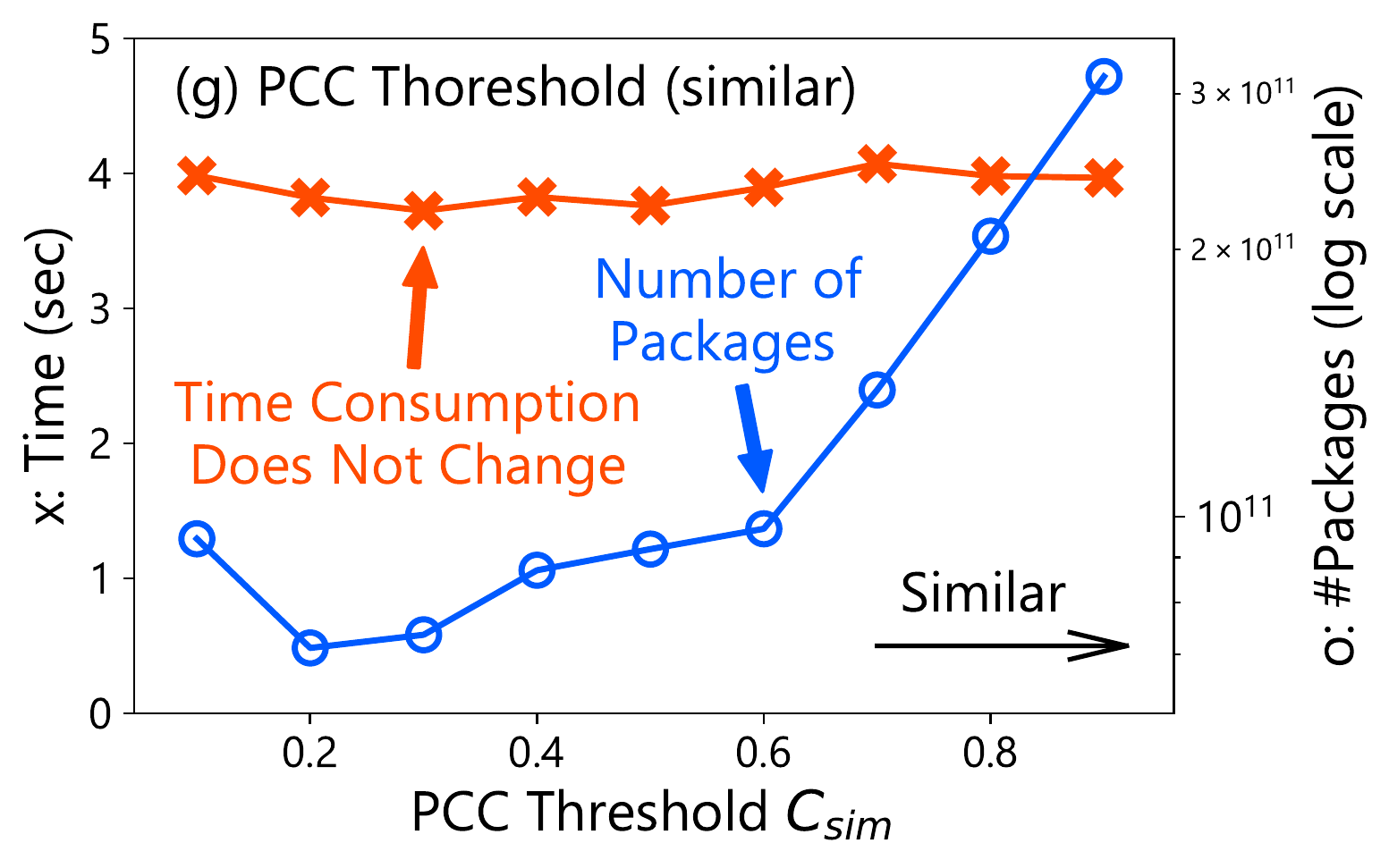}
\end{minipage}
\begin{minipage}{0.22\hsize}
\centering
\includegraphics[width=\hsize]{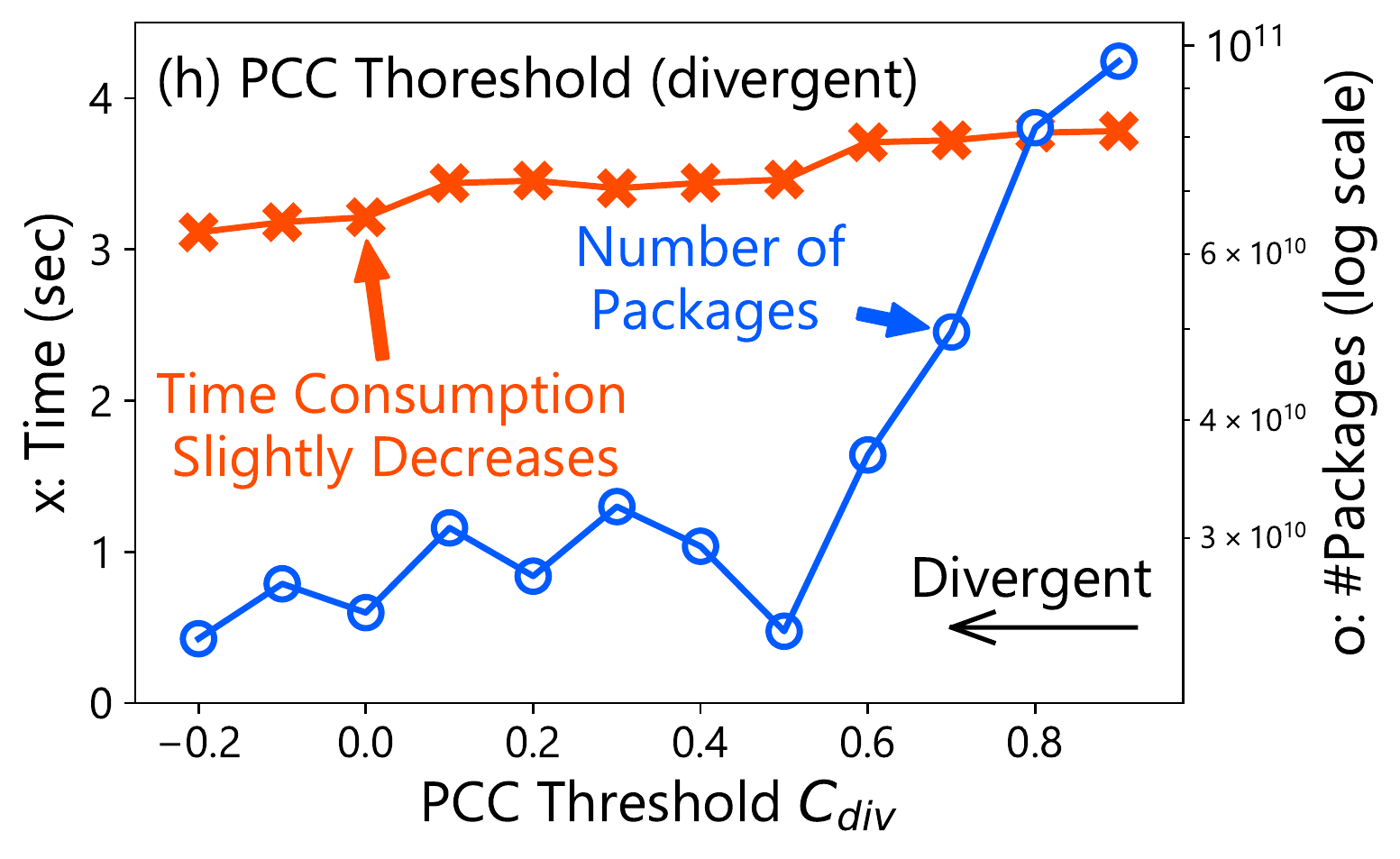}
\end{minipage}
\vspace{-0.12in}
\caption{Speed and number of qualified packages in various settings. The number of packages shown on the right axis is in a log scale. The time consumption shown on the left axis is in a linear scale, except for Figure (c).}
\label{fig: result}
\vspace{-0.12in}
\end{figure*}

\subsection{Experimental settings}

\uline{\textbf{Datasets.}} We use two standard datasets for recommendations.
\begin{itemize}
    \item \textbf{MovieLens} \cite{harper2016movielens}. We use five versions of this dataset, 100k, 1M, 10M, 20M, and 25M, to investigate the scalability of our method. An example scenario induced by this dataset is recommending a bundle of movies to a group of friends who are planning to see several movies together on their day off.
    \item \textbf{Amazon Home and Kitchen} \cite{he2016ups, mcauley2015image} contains reviews of home and kitchen products on amazon.com. To remove noisy items and users, we extract $10$-cores of the dataset, i.e., we iteratively remove items and users with less than $10$ interactions until all items and users have at least $10$ interactions. After preprocessing, the dataset contains $1395$ users, $1171$ items, and $25445$ interactions in total. This dataset is an example of bundle recommendations in e-commerce \cite{zhu2014bundle}.
\end{itemize}

\noindent \uline{\textbf{Group Generation.}} We use three strategies to create synthetic groups following \cite{kaya2020ensuring, serbos2017fairness}:
\begin{itemize}
    \item \textbf{Random:} Members are sampled uniformly at random without replacement.
    \item \textbf{Similar:} We sample members one by one. The first member is sampled uniformly at random. The following members are sampled uniformly at random from the set of users whose Pearson correlation coefficients (PCCs) of the rating vector to the already selected members are at least $C_\text{sim}$, where $C_\text{sim} \in [-1, 1]$ is a threshold hyperparameter. The generated groups tend to have similar preferences.
    \item \textbf{Divergent:} The sampling process is the same as the similar strategy except that the candidate members are users whose minimum PCC to the already selected members is at most $C_\text{div}$, where $C_\text{div} \in [-1, 1]$ is a threshold hyperparameter. The generated groups tend to have divergent preferences.
\end{itemize}

\noindent \uline{\textbf{Baselines.}} We use seven fair package recommendation methods.

\begin{itemize}
    \item \textbf{AveRanking \cite{baltrunas2010group}:} Items are ranked by average preferences. This algorithm maximizes the preference exactly but does not take fairness into consideration.
    \item \textbf{LMRanking \cite{baltrunas2010group}:} Items are ranked by the minimum preferences in the group.
    \item \textbf{GreedyVar \cite{Lin2017fairness}:} Items are chosen greedily so that the balanced score of the preferences and fairness scores is maximized. This method uses variance as the fairness score. We set $\lambda = 0.5$ (i.e., balancing completely).
    \item \textbf{LMRanking \cite{Lin2017fairness}:} Items are chosen greedily so that the balanced score of the preferences and fairness scores is maximized. This method uses the least misery as the fairness score. We set $\lambda = 0.5$ (i.e., balancing completely).
    \item \textbf{GFAR \cite{kaya2020ensuring}:} Items are chosen greedily so that the sum of the relevance scores in the group is maximized.
    \item \textbf{SPGreedy \cite{serbos2017fairness}:} Items are chosen greedily with respect to the proportionality. Although this is not exact, it showed excellent performance in proportionality in the original paper.
    \item \textbf{EFGreedy \cite{serbos2017fairness}:} Items are chosen greedily with respect to envy-freeness. Although this is not exact, it showed excellent performance in envy-freeness in the original paper.
\end{itemize}

\subsection{Speed and Number of Packages} \label{sec: result}

\begin{table*}[t]
    \centering
    \caption{Balancing scores. Proportionality, envy-freeness, and preference are normalized to be in $[0, 1]$. The total score is the sum of these values. Thus, the maximum total score is three. Our method achieves the best proportionality and envy-fairness, whereas its preferences are second or third best. Highest values are marked in \textbf{bold} in each column.}
    \vspace{-0.1in}
    \scalebox{0.8}{
    \begin{tabular}{c|cccc|cccc} \toprule
        & \multicolumn{4}{c|}{MovieLens1M} & \multicolumn{4}{c}{Amazon} \\
        & \multicolumn{1}{c}{Proportionality} & \multicolumn{1}{c}{Envy-freeness} & \multicolumn{1}{c}{Preference} & \multicolumn{1}{c|}{Total Score} & \multicolumn{1}{c}{Proportionality} & \multicolumn{1}{c}{Envy-freeness} & \multicolumn{1}{c}{Preference} & \multicolumn{1}{c}{Total Score} \\ \midrule
        AveRanking \cite{baltrunas2010group} & \bft{1.000 $\pm$ 0.000} & 0.725 $\pm$ 0.156       & \bft{0.911 $\pm$ 0.027} & 2.636 $\pm$ 0.171       & \bft{1.000 $\pm$ 0.000} & 0.500 $\pm$ 0.125       & \bft{0.939 $\pm$ 0.012} & 2.439 $\pm$ 0.126 \\
        LMRanking \cite{baltrunas2010group}  & 0.988 $\pm$ 0.037       & 0.588 $\pm$ 0.168       & 0.876 $\pm$ 0.036       & 2.451 $\pm$ 0.188       & 0.912 $\pm$ 0.263       & 0.425 $\pm$ 0.139       & 0.924 $\pm$ 0.031       & 2.261 $\pm$ 0.373 \\
        GreedyVar \cite{Lin2017fairness}     & 0.912 $\pm$ 0.080       & 0.750 $\pm$ 0.112       & 0.812 $\pm$ 0.035       & 2.474 $\pm$ 0.157       & 0.787 $\pm$ 0.202       & 0.637 $\pm$ 0.088       & 0.859 $\pm$ 0.031       & 2.284 $\pm$ 0.287 \\
        GreedyLM \cite{Lin2017fairness}      & 0.950 $\pm$ 0.061       & 0.775 $\pm$ 0.109       & 0.813 $\pm$ 0.036       & 2.538 $\pm$ 0.155       & 0.662 $\pm$ 0.159       & 0.600 $\pm$ 0.094       & 0.853 $\pm$ 0.031       & 2.115 $\pm$ 0.249 \\
        GFAR \cite{kaya2020ensuring}         & 0.950 $\pm$ 0.061       & 0.762 $\pm$ 0.104       & 0.812 $\pm$ 0.038       & 2.525 $\pm$ 0.154       & 0.762 $\pm$ 0.142       & 0.650 $\pm$ 0.075       & 0.871 $\pm$ 0.025       & 2.284 $\pm$ 0.219 \\
        SPGreedy \cite{serbos2017fairness}   & \bft{1.000 $\pm$ 0.000} & 0.525 $\pm$ 0.156       & 0.851 $\pm$ 0.041       & 2.376 $\pm$ 0.167       & \bft{1.000 $\pm$ 0.000} & 0.375 $\pm$ 0.079       & 0.867 $\pm$ 0.015       & 2.242 $\pm$ 0.085 \\
        EFGreedy \cite{serbos2017fairness}   & 0.925 $\pm$ 0.127       & \bft{1.000 $\pm$ 0.000} & 0.792 $\pm$ 0.053       & 2.717 $\pm$ 0.165       & 0.750 $\pm$ 0.244       & 0.838 $\pm$ 0.080       & 0.854 $\pm$ 0.027       & 2.441 $\pm$ 0.302 \\
        Ours                                 & \bft{1.000 $\pm$ 0.000} & \bft{1.000 $\pm$ 0.000} & 0.888 $\pm$ 0.037       & \bft{2.888 $\pm$ 0.037} & \bft{1.000 $\pm$ 0.000} & \bft{0.912 $\pm$ 0.057} & 0.913 $\pm$ 0.020       & \bft{2.825 $\pm$ 0.064} \\ \midrule
        Ours ($10$th)                        & \bft{1.000 $\pm$ 0.000} & \bft{1.000 $\pm$ 0.000} & 0.887 $\pm$ 0.037       & 2.887 $\pm$ 0.037       & \bft{1.000 $\pm$ 0.000} & \bft{0.912 $\pm$ 0.057} & 0.911 $\pm$ 0.020       & 2.824 $\pm$ 0.064 \\
        Ours ($100$th)                       & \bft{1.000 $\pm$ 0.000} & \bft{1.000 $\pm$ 0.000} & 0.881 $\pm$ 0.040       & 2.881 $\pm$ 0.040       & \bft{1.000 $\pm$ 0.000} & 0.900 $\pm$ 0.050       & 0.905 $\pm$ 0.025       & 2.805 $\pm$ 0.058 \\
        Ours (random)                        & \bft{1.000 $\pm$ 0.000} & \bft{1.000 $\pm$ 0.000} & 0.720 $\pm$ 0.044       & 2.720 $\pm$ 0.044       & \bft{1.000 $\pm$ 0.000} & \bft{0.912 $\pm$ 0.057} & 0.862 $\pm$ 0.031       & 2.774 $\pm$ 0.080 \\
        \bottomrule
    \end{tabular}
    }
    \label{tab: tradeoff}
    \vspace{-0.12in}
\end{table*}
We investigate the relations among the speed of \textsc{FAPE}, the number of packages, and various parameters of the problem, including the dataset size, package size $K$, group size $g$, threshold $\Delta$, threshold $\tau$, fairness criteria, and grouping criteria. We set $g = 8$, $K = 4$, $\Delta = 5$ for proportionality, $\Delta/100 = 2/g$ for envy-freeness, $\tau = g$, the MovieLens1M dataset, proportionality criterion, and random grouping as the default setting and vary each aspect in the following analysis. We use a Linux server with an Intel Xeon CPU E7-4830 @ 2.00GHz and 1TB RAM. The time consumption reported in the next subsection is evaluated on one core of the CPU.  Figure \ref{fig: result} shows the results. Note that the number of packages is plotted on a log scale because the combination of packages can be exponentially large. We plot time consumption on a linear scale, except for the group size experiment, where an exponential time is expected due to Theorem \ref{thm: enum}. The analysis of the results is discussed below.

\noindent \textbf{Dataset Size} (Figure \ref{fig: result} (a)): We use datasets of various sizes: MovieLens100k, 1M, 10M, 20M, and 25M. They contain $1682$, $3706$, $10677$, $26744$, and $59047$ items, respectively. Although the problem is NP-hard, the time consumption of \textsc{FAPE} scales only \emph{linearly}, and it scales to large datasets owing to the fixed parameter tractability. This result is consistent with Theorem \ref{thm: enum}. It can also be seen that the number of completely fair packages, where all members are satisfied, is unexpectedly large (e.g., > $10^8$ packages for MovieLens100k). The existing algorithms (e.g., SPGreedy \cite{serbos2017fairness}) optimize scores by some procedures and output only one package from as many as $10^8$ candidates. It is not clear which package is chosen by these methods as a consequence of the optimization procedures. By contrast, our method enables us to choose packages in a way each user specifies using filtering operations (see Sections \ref{sec: op} and \ref{sec: balance}).

\noindent \textbf{Package Size} (Figure \ref{fig: result} (b)): We vary the package size $K$ from $2$ to $8$. The number of packages increases exponentially with respect to $K$ because $O(n^K)$ combinations of items exist. Nonetheless, the time consumption scales \emph{linearly}, which is consistent with Theorem \ref{thm: enum}.

\noindent \textbf{Group Size} (Figure \ref{fig: result} (c)): We vary the size $g$ of groups from $2$ to $16$. Note that the time axis is on a log scale only for this experiment. As expected from Theorem \ref{thm: enum}, the time consumption increases exponentially. This is the main limitation of our proposed algorithm. However, we stress that many applications in the literature consider at most $8$ members (e.g., families). Thus, our algorithm is tractable in these cases. We observe that the number of qualified packages decreases because, with large groups, it is difficult to satisfy all.

\noindent \textbf{Proportionality Threshold} (Figure \ref{fig: result} (d)): We vary the threshold $\Delta$ that determines the set of items that each user likes. We observe that the number of qualified packages converges to $\binom{n}{K}$ because all packages are fair if all members like all items. It can also be observed that the time consumption increases first and then decreases. This indicates that the intermediate case is the most difficult. Overall, our algorithm is efficient in all regions.

\noindent \textbf{Envy-freeness Threshold} (Figure \ref{fig: result} (e)): We use envy-freeness instead of proportionality as the fairness criterion. We vary the threshold $\Delta$. Note that $\Delta/100 = k/g$ means that a member is envy-free for item $i$ if his/her rating for item $i$ is in top-$k$ in the group. Similar to proportionality, the time consumption increases and then decreases, and the number of packages converges to $\binom{n}{K}$.

\noindent \textbf{Threshold} (Figure \ref{fig: result} (f)): We vary the threshold $\tau$ of proportionality from $1$ to $8 (= g)$. As the constraint becomes more stringent, the number of qualified packages decreases.

\noindent \textbf{Grouping} (Figures \ref{fig: result} (g) and (h)): We use similar and divergent grouping instead of random grouping. The number of packages increases as the group becomes homogeneous. The time consumption does not change much regardless of the properties of the groups.

Overall, our algorithm is efficient in various settings. It enumerates as many as $10^{12}$ items within a few seconds.

\subsection{Balancing Fairness and Preference} \label{sec: balance}

Ensuring fairness is not the only requirement in package-to-group recommendations. Even if a package is envy-free, it is useless if all members dislike items in the package. Balancing fairness and preference is important. In addition, both proportionality and envy-freeness may be simultaneously required because they model different aspects of fairness. We show that our method achieves high proportionality, envy-freeness, and preference simultaneously by the intersection and optimization operations introduced in Section \ref{sec: op}. The objective score is the sum of proportionality, envy-freeness, and total preference $\mathcal{S}(\mathcal{P})$. We call it the total score. We divide the proportionality and the envy-freeness by $g$ and normalize the preference to ensure that each term is within $[0, 1]$.

First, we run \textsc{FAPE} with each threshold $\tau \in \{1, 2, \cdots, g\}$ and build ZDDs for proportionality and envy-freeness.
For each $\tau$ and $\tau'$, we build the ZDD $\mathfrak{I}_{\tau, \tau'}$ with at least $\tau$ proportionality and $\tau'$ envy-freeness by the intersection operation. Then, we draw the package with the maximum total preference from each ZDD by the optimization operation. We adopt the package with the maximum objective score among $\mathfrak{I}_{\tau, \tau'}$. This procedure outputs the \emph{exactly} optimal package with respect to the total score in a reasonable time.

We compare the performance of our method and existing methods. We set $g = 8$, $K = 4$, $\Delta = 5$ for proportionality, and $\Delta/100 = 2/g$ for envy-freeness. Table \ref{tab: tradeoff} reports the average scores and standard deviations for randomly created $10$ groups. It shows that our method achieves the best performance in both proportionality and envy-freeness simultaneously. It is remarkable that on the Amazon dataset, the envy-freeness of our method is better than that of EFGreedy, which is designed to optimize envy-freeness. Although AveRanking, which optimizes preference exactly, performs better than our method in preference, it performs poorly in envy-freeness and may frustrate some members. Our method computes packages with slightly worse preference but much better fairness than AveRanking. Overall, our method outperforms the other methods with respect to the total score. It indicates that our method strikes a better balance of fairness and preference than the existing methods. 

We also compute $10$-th and $100$-th best packages and random packages from the best ZDD (i.e., the ZDD that contains the best package) using the optimization and sampling operations. Recall that even random packages are satisfactory because the ZDD contains only fair packages. The bottom three rows in Table \ref{tab: tradeoff} report the results. These results indicate that our method provides many effective candidates, whereas the existing fair package-to-group methods provide only one candidate.

\section{Related Work}

Fairness is a crucial aspect in many applications \cite{kamishima2012enhancement, yao2017beyond, bruke2017mltisided, hardt2016equality, zafar2017fairness}. In this work, we study a special case of recommendation tasks, package-to-group recommendations \cite{qi2016recommending, serbos2017fairness}. In group recommendations, items are recommended to a group instead of an individual \cite{jameson2007recommendation}. \citet{serbos2017fairness} and others \cite{qi2016recommending, xiao2019beyond, kaya2020ensuring} found that group recommender systems were also biased against some members in a group. In packages recommendations (or composite recommendations \cite{xie2010breaking} or bundle recommendations \cite{zhu2014bundle}), a set of items forms a unified package toward a single common goal \cite{xie2010breaking, yahia2014composite, zhu2014bundle, xie2014generating}. As \citet{xie2010breaking} pointed out, recommending top-$K$ packages (thus, a list of sets of items) is crucial to provide choices of packages to users. \citet{deng2012complexity} proved the complexities of package recommendations. Their results are fundamentally different from ours in two aspects. First, we consider fair and group recommendations which were not considered by them. Second, we derived FPT algorithms for practical implementations, whereas their interest was in theoretical aspects. Although the complexities of fair item allocation problems have been studied \cite{bouveret2008efficiency, bliem2016complexity}, they are not in the context of recommendations nor enumeration.

The most related works are \cite{qi2016recommending, serbos2017fairness}, where fair package-to-group recommender systems were proposed. There are two crucial differences between this work and theirs. First, we \emph{enumerate} all fair packages, whereas they output a single package as a result of optimization. Second, our algorithm does not resort to any approximations. Thus, our method provably performs better than their methods while keeping efficiency owing to its fixed time tractability. 

\vspace{0.05in}
\noindent \uline{\textbf{Difference with Maximization Algorithms.}} We point out several crucial differences between enumeration and maximization algorithms. One may think that maximization algorithms, such as SPGreedy and EFGreedy, can also enumerate packages by creating solutions and removing them from the search space iteratively. However, greedy algorithms just avoid exact matches, change only the last item, and fail to provide diverse packages. Besides, it is not obvious how to retrieve more than one packages using ranking-based algorithms, such as AveRanking and LMRanking. Furthermore, it is infeasible to enumerate all fair packages by generating them one by one because there are as many as $10^{14}$ candidates. By contrast, our algorithm can enumerate all packages all at once and enables filtering favorite packages by a variety of operations.

\section{Conclusion}

In this paper, we investigated the fair package-to-group recommendation problem. We proposed \emph{enumerating} all fair packages instead of computing a single package. Although the enumeration problem is computationally challenging, we proved that it is FPT with respect to the group size and proposed an efficient algorithm based on ZDDs. Our proposed algorithm can not only enumerate packages but also filter items by the intersection operation, optimize preferences by the linear Boolean programming, and sample packages uniformly and randomly. We experimentally confirmed that our proposed method scales to large datasets and can enumerate as many as one trillion packages in a reasonable time.


\begin{acks}
This work was supported by JSPS KAKENHI GrantNumber 21J22490.
\end{acks}

\bibliographystyle{plainnat}
\bibliography{sample-base}










\end{document}